%% file: arxiv.tex
\documentclass[a4paper,oneside,11pt]{article}
\usepackage{a4wide}
\usepackage{hyperref}
\hypersetup{
    bookmarks=true,         
    unicode=false,          
    pdftoolbar=true,        
    pdfmenubar=true,        
    pdffitwindow=false,     
    pdfstartview={FitH},    
    colorlinks=true,       
    linkcolor=blue,          
    citecolor=red,        
    filecolor=magenta,      
    urlcolor=cyan           
}
\usepackage[utf8]{inputenc}

\usepackage{comment}
\usepackage{amsmath,amssymb,makeidx,mathrsfs}
\usepackage[english,ruled,linesnumbered]{algorithm2e}
\usepackage{graphicx}
\newtheorem{theorem}{Theorem}[section]
\newtheorem{lemma}{Lemma}[section]
\newtheorem{corollary}{Corollary}[section]
\newtheorem{definition}{Definition}[section]
\newenvironment{proof}{\noindent\emph{Proof:}}{\hfill$\Box$}

\newcommand{\N}{\mathbb{N}}
\newcommand{\R}{\mathbb{R}}
\newcommand{\Prob}{\mathbb{P}}
\newcommand{\E}{\mathbb{E}}

\newcommand{\A}{\mathcal{A}}
\newcommand{\F}{\mathcal{F}}
\newcommand{\ie}{\textit{ie}}
\newcommand{\SA}{\mathcal{SA}}
\newcommand{\composition}[0]{\circ}

\begin{document}

\title{The Random Bit Complexity \\ of Mobile Robots Scattering}

\author{Quentin Bramas$^1$ and S\'{e}bastien Tixeuil$^{1,2}$}

\maketitle
\begin{center}
$^1$Sorbonne Universités, UPMC Univ Paris 06, UMR 7606, F-75005, Paris, France\\
$^2$Institut Universitaire de France
\end{center}

\begin{abstract}
We consider the problem of scattering $n$ robots in a two dimensional continuous space. As this problem is impossible to solve in a deterministic manner~\cite{DieudonneP09}, all solutions must be probabilistic. We investigate the amount of randomness (that is, the number of random bits used by the robots) that is required to achieve scattering.

We first prove that $n \log n$ random bits are necessary to scatter $n$ robots in any setting. Also, we give a sufficient condition for a scattering algorithm to be random bit optimal. As it turns out that previous solutions for scattering satisfy our condition, they are hence proved random bit optimal for the scattering problem.


Then, we investigate the time complexity of scattering when strong multiplicity detection is not available. We prove that such algorithms cannot converge in constant time in the general case and in $o(\log \log n)$ rounds for random bits optimal scattering algorithms. 
However, we present a family of scattering algorithms that converge as fast as needed without using multiplicity detection. Also, we put forward a specific  protocol of this family that is random bit optimal ($n \log n$ random bits are used) and time optimal ($\log \log n$ rounds are used). This improves the time complexity of previous results in the same setting by a $\log n$ factor.

Aside from characterizing the random bit complexity of mobile robot scattering, our study also closes its time complexity gap with and without strong multiplicity detection (that is, $O(1)$ time complexity is only achievable when strong multiplicity detection is available, and it is possible to approach it as needed otherwise).
\end{abstract}
\paragraph{keyword}
    Mobile robots, Scattering, Probabilistic algorithms, Complexity

\section{Introduction}
We consider distributed systems consisting of multiple autonomous robots ~\cite{2012Flocchini,SuzukiY99} that can move freely on a 2-dimensional plane, observe their surroundings and perform computations. The robots do not communicate explicitly with other robots and there is no central authority that communicates with the robots. Such teams of robots can be deployed in areas inaccessible to humans, to perform collaborative tasks such as search and rescue operations, data collections, environmental monitoring and even extra-terrestrial exploration. From the theoretical point of view, the interest lies in determining which tasks can be performed by such robot teams and under what conditions. 

One line of research is to determine the minimum capabilities required by the robots to achieve any given task~\cite{2012Flocchini}. A particularly weak model of robots is assumed and additional capabilities are added whenever it is necessary to solve the problem. In our model, the robots are assumed to be anonymous (\emph{i.e.} indistinguishable from one another), oblivious (\emph{i.e.} no persistent memory of the past is available) and disoriented (\emph{i.e.} they do not agree on a common coordinate system nor a common chirality). The robots operate in Look-Compute-Move (LCM) cycles, where in each cycle a robot \emph{Looks} at its surroundings and obtains a snapshot containing the locations of all robots as points on the plane with respect to its own location and ego-centered coordinate system; Based on this visual information, the robot \emph{Computes} a destination location and then \emph{Moves} towards the computed location. Since the robots are identical, they all follow the same algorithm. The algorithm is oblivious if the computed destination in each cycle depends only on the snapshot obtained in the current cycle (and \emph{not} on the past history of execution). The snapshots obtained by the robots are not consistently oriented in any manner.

When processing a snapshot, a robot can distinguish whether a point is empty (\emph{i.e.}, not occupied by any robot). However, since robots are viewed as points, the question arises of how robots occupying the same position at the same time are perceived in a snapshot. The answer to this question is formulated in terms of the capacity of the robots to detect multiplicity of robots in a point. The robots are said to be capable of \emph{multiplicity detection} if they can distinguish if a point is occupied by one or more than one robot. 

One important task useful in multi-robot coordination is \emph{gathering} the robots at a single location, not known beforehand. The dual problem of gathering is the \emph{scattering} problem. Scattering requires that, starting from an arbitrary configuration, eventually no two robots share the same location. It turns out that neither deterministic gathering \cite{SuzukiY99} nor scattering \cite{DieudonneP09} are possible without additional assumptions. Most of the work done so far in order to circumvent the impossibility of gathering focuses on required minimal additional assumptions with respect to the coordinate system or multiplicity detection~\cite{2012Flocchini,SuzukiY99} to make the problem solvable. However, the scattering problem cannot allow deterministic solutions \cite{DieudonneP09}. 

\paragraph{Related Work.} The first probabilistic algorithms to solve mobile robot scattering without multiplicity detection were given by Dieudonné and Petit~\cite{DP07c,DieudonneP09}. The algorithms are based on the following simple scheme: after the Look phase, a robot computes the Voronoi diagram \cite{aurenhammer1991voronoi} of the observed positions, and then tosses a coin ($\frac{1}{4}$~\cite{DP07c} or $\frac{1}{2}$~\cite{DieudonneP09}) to either remain in position, or move toward an arbitrary position in its Voronoi cell. The fact that a robot may only move within its Voronoi cell preserves the fact that initially distinct robots (that is robots occupying distinct positions) remain distinct thereafter. This invariant and the positive probability that two robots on the same point separate implies the eventual scattering of all robots. A later study~\cite{ClementDPIM10} shows that the scattering algorithm~\cite{DieudonneP09} converges in expected $O(\log n \log\log n)$ rounds. 
In the same paper~\cite{ClementDPIM10}, a new probabilistic algorithm was presented, with the assumption that robots are aware of the total number of robots. This protocol is optimal in time as it scatters any $n$-robots configuration in expected $O(1)$ rounds. If the total number of robots $n$ is known, then robots are able to choose uniformly at random a position within their Voronoi cell among $2n^2$ possibilities, inducing an expected $O(1)$ rounds scattering time. In the limited visibility setting~\cite{IPT10c} (the visibility capability of each robot has a constant radius, and visual connectivity has to be maintained throughout scattering), the time lower bound grows to expected $n$ rounds for scattering $n$ robots. None of the aforementioned works investigated the number of random bits used in the scattering process.

\input{link-with-naming-problem.tex}

\paragraph{Our contribution.} 
We investigate the amount of randomness (that is, the number of random bits used by the robots) that is necessary to achieve mobile robots scattering. In more details, we first define a canonical scattering algorithm, that encompasses all previous solutions, and is tantamount to selecting the number of possible locations that are selected uniformly at random by the robots.

Then, we prove that $n \log n$ random bits are necessary to scatter $n$ robots in any setting for all scattering algorithms (not only canonical algorithm). Also, we give a sufficient condition for a canonical scattering algorithm to be random bit optimal (namely, the number of possible locations must be polynomial in the number of observed positions). As it turns out that previous solutions for scattering~\cite{ClementDPIM10,DP07c,DieudonneP09} satisfy our condition, they are hence proved random bit optimal for the scattering problem.

Finally, we investigate the time complexity 
of scattering algorithms, when strong multiplicity detection is not available. We prove that such algorithms cannot converge in constant time in the general case and in $o(\log \log n)$ rounds in the case of random bits optimal algorithms (in this last setting, the best known upper bound was $\log n \log \log n$~\cite{DP07c,DieudonneP09}).
On the positive side, we provide a family of scattering algorithms that converge as fast (but not $O(1)$) as needed, without using multiplicity detection. Also, we give a particular protocol among this family that is random bit optimal ($n \log n$ random bits are used) and time optimal ($\log \log n$ rounds are used). This improves the time complexity of previous results in the same setting by an expected $\log n$ factor.



\section{Model and Preliminaries}
\paragraph{Robot networks.} There are $n$ robots modeled as points on a geometric
plane. A robot can observe its environment and determine
the location of other robots in the plane, relative to its
own location and coordinate system. All robots are identical (and thus
indistinguishable) and they follow the same algorithm.
Moreover, each robot has its own local coordinates system,
which may be distinct from that of other robots.
In this paper, robots are said to have unlimited visibility, in the sense that they are always able to sense the
position of all other robots, regardless of their proximity.

\paragraph{Multiplicity detection.} When several robots share the same location, this location is called a point of multiplicity.
Robots are capable of \emph{strong} multiplicity detection when they are aware of the number of robots located at
each point of multiplicity. In contrast, when robots are capable of weak multiplicity detection, they know which points are
points of multiplicity, but are unable to count how many robots are located there. 
The multiplicity detection of a robot is said to be \emph{local} if the multiplicity detection concern only the point where robot lies. If robots detect the multiplicity of each observed point, the multiplicity detection is \emph{global}. Robots are not aware of the actual number $n$ of robots unless they are capable of global strong multiplicity detection. 
%
If robots are not able to detect multiplicity, they never know if the configuration is scattered and thus never stop moving. Hence, algorithms that do not use multiplicity detection cannot terminate. With local weak multiplicity detection robots are aware of the situation at their position, \emph{e.g.} they can stop executing the algorithm if they sense they are alone at their location. However, they may not know if the global configuration is scattered (yet, if the configuration is indeed scattered, all robots are stopped and the algorithm (implicitly) terminates). With global weak multiplicity detection, algorithm can explicitly terminates when every observed position is not a multiplicity point.

\paragraph{System model.} 
Three different scheduling assumptions have been considered in previous work. The strongest model is the fully synchronous (FSYNC) model where each phase of each cycle is performed simultaneously by all robots. On the other hand, the weakest model, called asynchronous (ASYNC) allows arbitrary delays between the Look, Compute, and Move phases and the movement itself may take an arbitrary amount of time~\cite{2012Flocchini}. The semi-synchronous (SSYNC) model \cite{2012Flocchini,SuzukiY99} lies somewhere between the two extreme models. In the SSYNC model, time is discretized and at each considered step an arbitrary subset of the robots are active. The robots that are active, perform exactly one \emph{atomic} Look-Compute-Move cycle. It is assumed that a hypothetical scheduler (seen as an adversary) chooses which robots should be active at any particular time and the only restriction of the scheduler is that it must activate each robot infinitely often in any infinite execution (that is, the scheduler is \emph{fair}).

In this paper, for the analysis, we use the FSYNC model. Lower bounds naturally extend to SSYNC and ASYNC models and upper bounds (that is, algorithms) are also valid in SSYNC. Indeed, as in \cite{ClementDPIM10}, since all the algorithms in this paper ensure that two robots moving at different times necessarily have different destinations, the worst
case scenario is when robots are activated simultaneously.
In FSYNC model, robots perform simultaneously an \emph{atomic computational cycle} composed of the following three phases: Look, Compute, and Move.
\begin{itemize}
\item[$\bullet$] \emph{Look.} An observation returns a snapshot of the positions of all robots. All robots observe the exact same environment (according to their respective coordinate
systems).
\item[$\bullet$] \emph{Compute.} Using the observed environment, a robot executes its algorithm to compute a destination.
\item[$\bullet$] \emph{Move.} The robot moves towards its destination (by a non-zero distance but without always reaching it).
\end{itemize}
Moreover, robots are assumed to be oblivious (\emph{i.e.}, stateless), in the sense that a robot does not keep any
information between two different computational cycles. We evaluate the time complexity of algorithms using the number of asynchronous rounds required to scatter all robots. An asynchronous round is defined as the shortest fragment of an execution where each robot executes its cycle at least once.


\paragraph{Notations.}
In the sequel, $C$ denotes a $n$-robots configuration, that is, a multi-set containing the position of all robots in the plane. Removing multiplicity information (that is, multiple entries for the same position) from $C$ yields the corresponding set $U(C)$. For a multi-set $C$, $|C|$ denotes its cardinality. For a particular point $P\in\R^2$, $|P|$ denotes the multiplicity of $P$.
We denote by $\mathcal{C}(k,n)$ the set of $2$-tuples $(C,P)$ where $C$ is a $n$-robots configuration that contains a point $P$ of multiplicity $k$. 

%
%

\paragraph{Random Bits Complexity.} 
The number of random bits needed by a robot to choose randomly a destination among $k$ possible locations is at least $\log_2(k)$, regardless of the distribution, as long as each destination has a non-zero probability to be chosen. We denote by $\log = \log_2$ the logarithm with respect to base 2 obtained from the natural logarithm as $\log(x) = \frac{\ln(x)}{\ln(2)}$. Of course, since there is a probabilistic process involved, starting from the same initial configuration, the exact number of random bits may not be the same for two particular executions of a protocol. So, in the sequel, we consider the expected number of random bits used for scattering. 

Since we are concerned about the scattering problem, we do not take into account random bits used by robots that are not located at a point of multiplicity (\emph{i.e}, robots that are already scattered). Of, course, all our lower bound results remain valid without this assumption, but upper bounds we provide do make use of this hypothesis when robots are not capable of weak local multiplicity (as termination cannot be insured in this case). We also assume that robots cannot use an infinite number of random bits in a single execution. 

For the study of the random bits complexity, we define $Z_{C,P}$, the random variable that represents the number of random bits used by an algorithm to scatter the robots in $P$ starting from the configuration $C$. Formally, for an algorithm $\A$, $Z_{C,P}$ is defined over all the possible executions of $\A$ (starting with the configuration $C$ that contains the point $P$). For an execution, $Z_{C,P}$ equals $b$ if and only if the number of random bits used to scatter all robots that are initially in $P$ is $b$ (ignoring the robots in $C$ that are not initially located at $P$).

For a point $P$ of multiplicity $n$, we can represent the way robots at $P$ are divided over $k$ possible destinations with a multi-index $\alpha\in\N^k$ such that $|\alpha|=\sum_{i=1}^k\alpha_i = n$. The resulting maximum multiplicity is denoted by $\|\alpha\|_{\infty} = \max_i\alpha_i$. Consider the random variable $X$ that equals $\alpha\in\N^k$ if and only if the robots in $P$ are divided in $k$ points of multiplicity $\alpha_1$, $\alpha_2$, $\ldots$ and $\alpha_k$.

It is known that:
\begin{equation}
\mathbb{E}(Z_{C,P}) = \sum_{\alpha\in\mathbb{N}^k,\; |\alpha|=n}
\Prob(X=\alpha)\mathbb{E}(Z_{C,P} | X = \alpha)    
\end{equation}
Then, $\mathbb{E}(Z_{C,P} | X = \alpha)$ equals the number of random bits used during the first round ($n\log(k)$) plus the expected number of random bits used to scatter the $k$ points $p_1$, $p_2$, \ldots and $p_k$, coming from $P$ of multiplicity $\alpha_1$, $\alpha_2$, $\ldots$ and $\alpha_k$. Of course the rest of the configuration may have changed too. But since we want to bound the expectation, we can have an upper or a lower bound by taking the worst or the best resulting configuration.
For all $N\in\N,\;n\leq N$, let 
\begin{equation}\label{eq: B an W definition}
B(n,N) = \min_{(C,P)\in\mathcal{C}(n,N)}(\E(Z_{C,P}))\quad\text{and}\quad W(n,N) = \max_{(C,P)\in\mathcal{C}(n,N)}(\E(Z_{C,P}))
\end{equation}
The existence of such $\min$ and $\max$ comes from the fact that for $N\in\mathbb{N}^*,\;n<N$, the set $\mathcal{C}(n,N)$, is finite. This is due to the fact that there exists an initial configuration from which some (deterministically computed but randomly chosen) paths have been followed by the robots. 

 Moreover, if Algorithm $\mathcal{A}$ makes sure that two robots at distinct locations in a given configuration remain at distinct locations thereafter, then for two distinct points $P$ and $P'$, $Z_{C,P}$ and $Z_{C,P'}$ are independent and their sum is exactly the number of random bits used to scatter $P$ and $P'$.
Then,
\begin{equation}\label{eq:recursive lower bound for RBC}
B(n,N) \geq n\log(k) +
 \sum_{\alpha\in\mathbb{N}^k,\quad
 |\alpha|=n} \Prob(X=\alpha)\sum_{i=1}^kB(\alpha_i,N)
\end{equation}
\begin{equation}\label{eq:recursive upper bound for RBC}
W(n,N) \leq n\log(k) +
 \sum_{\alpha\in\mathbb{N}^k,\quad
 |\alpha|=n} \Prob(X=\alpha)\sum_{i=1}^kW(\alpha_i,N)
\end{equation}
The recursive inequality (\ref{eq:recursive lower bound for RBC}) is used in lemma \ref{lem:lower bound for RBC} to find the lower bound, and the recursive inequality (\ref{eq:recursive upper bound for RBC}) in Theorem~\ref{theorem-optimality-caracterization} to find the upper bound.

\paragraph{A Canonical Scattering Algorithm.}\label{sec:canonical scattering algorithm}

Let $\mathcal{A}$ be a scattering algorithm. As $\mathcal{A}$ can't be deterministic \cite{DieudonneP09}, the computation of the location to go to must result from a probabilistic choice (more practically, a robot must randomly choose a destination among a previously computed set of possible destinations). We note $k_\mathcal{A}(C,P)$ the function that returns the number of possible destinations depending on the current observed (global) configuration $C$, and the current observed (local) point $P$ (that is, the point where the robot executing the algorithm lies). Robots located at $P$ may not be aware of $P$'s multiplicity, but they base their computation of the possible destinations set on the same observation. We assume an adversarial setting where symmetry is preserved unless probabilistic choices are made, so we expect the local coordinate systems of all robots occupying the same position $P$ to be identical. Thus, the set of possible destinations is the same for all robots at $P$. We now define a canonical scattering algorithm that generalizes previously known scattering algorithms.

\begin{definition}An algorithm $\mathcal{A}$ is a \emph{canonical} scattering algorithm if it has the following form:\\
\begin{algorithm}[H]\label{algo-generique}
$C \leftarrow$ Observed current configuration.\\
$P \leftarrow$ Observed current position of $r$.\\
Compute a set of $k_\mathcal{A}(C,P)$ possible destinations $Pos$ such that every point in $Pos$ may not be chosen by a robot not currently in $P$.\\
Move toward a point in $Pos$ chosen uniformly at random.\\
\caption{Canonical scattering algorithm, executed by a robot $r$}
\end{algorithm}
The function $k_\mathcal{A}$ that gives the number of possible destinations depending on the current configuration and position is called the \emph{destination function} of the algorithm $\mathcal{A}$.
\end{definition}

Line $3$ of Algorithm~\ref{algo-generique} implies that a canonical algorithm must ensure that the multiplicity of any given point never increases (\emph{i.e.} robots located at different locations remain at different locations thereafter). Previous algorithms~\cite{ClementDPIM10,DieudonneP09} are canonical in the SSYNC and FSYNC models. Both of them use Voronoi Diagrams  \cite{aurenhammer1991voronoi} to ensure monotonicity of multiplicity points. The algorithm given in \cite{DieudonneP09} is a canonical scattering algorithm with line $3$ replaced by: \textsf{Compute a set $Pos$ of $2$ points in the Voronoi cell of $r$}\label{algo-scattering-weak-multiplicity}. The algorithm given in \cite{ClementDPIM10} is a canonical scattering algorithm with line $3$ replaced by: \textsf{Compute a set $Pos$ of $2|C|^2$ points in the Voronoi cell of $r$}\label{algo-scattering-strong-multiplicity}. This property holds only if Voronoi cells computations occur at the same time (that is, in the FSYNC and the SSYNC models). In the ASYNC model, two robots at different positions, activated at different times, may move towards the same destination\footnote{To our knowledge, no algorithm exists for scattering mobile robots in the ASYNC model.}.


An algorithm that computes a set $Pos$ of $k$ points but does not choose uniformly the destination from $Pos$ can be seen as a canonical scattering algorithm if points in $Pos$ can have multiplicity greater than $1$, \emph{i.e.} if $Pos$ is a multi-set. For example, if an algorithm computes a set $Pos = \{x_1,x_2\}$ and chooses $x_1$ with probability $\frac{3}{4}$, this is equivalent to choosing uniformly at random from the multi-set$\{x_1,x_1,x_1,x_2\}$. Of course this scheme cannot be extended to irrational probability distributions, yet those distributions induce an infinite number of random bits, which is not allowed in  our model. We can now state our first Theorem:

\begin{theorem}
\label{thm:canonical}
If $\mathcal{A}$ is an algorithm that ensures that the multiplicity of any point never increases, then $\mathcal{A}$ is a canonical scattering algorithm (with $Pos$ possibly a multiset). 
\end{theorem}

Observe that any deterministic protocol for mobile robot networks (that ensure monotonicity of multiplicity points) can be seen as a canonical scattering algorithm whose destination function is identically $1$. Also, if an algorithm computes a multiset $Pos$ with duplicate positions, it uses more random bits to select its destination at any given stage of the computation. As we focus on efficient algorithms (that is, we try to minimize the number of random bits), we suppose from now on that $Pos$ is a set (\emph{i.e.} it has no duplicate positions). Indeed the uniform distribution is the probability distribution that has the largest entropy.

%
%

\section{The Random Bit Complexity of Scattering}\label{sec:Random Bit Complexity of Scattering}

In this section we demonstrate that any mobile robots scattering algorithm must use at least $n\log(n)$ random bits, whether or not it uses multiplicity detection. Then we prove a sufficient condition for a canonical scattering algorithm to effectively use $O(n \log n)$ random bits. As this condition is satisfied by previously known canonical scattering algorithms~\cite{ClementDPIM10,DP07c,DieudonneP09}, a direct consequence of our result is that those algorithms are random bit optimal. 

\subsection{Lower Bound.}

In this section we prove that the expected number of random bits used by any scattering algorithm is greater than $n\log(n)$. Actually, we implicitly prove that any execution of a scattering algorithm that scatters $n$ robots initially located at the same position uses more than $n\log(n)$ random bits. The proof first considers \emph{canonical} scattering algorithms, and later expands to \emph{arbitrary} scattering algorithms (that is, algorithms that may not insure that the multiplicity of any point never increases, see Theorem~\ref{thm:canonical}).

\begin{lemma}\label{lem:lower bound for RBC}
Let $\mathcal{A}$ be a canonical scattering algorithm (Algorithm \ref{algo-generique}). The expected number of random bits needed to scatter $n$ robots is at least $n\log(n)$.
\end{lemma}

The sketch of the proof is as follows.
We prove that any execution of the algorithm uses at least $n\log(n)$ random bits by mathematical induction on the number of robots located at a particular point. For the base case, we observe that $2$ robots located at the same point and executed simultaneously must both use at least $1$ random bit. So, $2$ robots are scattered with more than $2 = 2\log 2$ random bits.

To prove the induction step, we observe that the most favorable scenario is when, at each round, robots are uniformly distributed over all possible destinations. If we assume that points with multiplicity $m<n$ need more than $m\log(m)$ random bits to be scattered, then if $n$ robots are split among two points of multiplicity $m_1$ and $m_2$, the number of random bits used to scatter those two points (which is greater than $m_1\log m_1 + m_2\log m_2 $) is greater than $n\log(n/2)$. This result comes from the convexity of function $x\mapsto x\log x$.

\begin{proof}
    
Let $B(n)$ be the minimum expected number of random bits to scatter $n$ robots initially located at Point $P$. 
Formally, with $N\in\N$ a fixed number, we consider \\$B(n) = min_{(C,P) \in \mathcal{C}(n,N)}\mathbb{E}(Z_{C,P})$. With this definition $B(n)$ may depend on $N$, but our bound does not.

Let $N\in\mathbb{N}^*$. The set $\mathcal{C}(n,N)$ of $2$-tuples $(C,P)$ where $C$ is a $N$-robots configuration that contains a point $P$ of multiplicity $n$ is finite. So, we can consider $B(n) = min_{(C,P) \in \mathcal{C}(n,N)}\mathbb{E}(Z_{C,P})$. With this definition $B(n)$ may depend on $N$ but we explain in the sequel that it actually does not.

During a round, each robot in $P$ choose randomly among a set of $k$ points. The number $k$ is the same for all the robots because they share the same location, execute the same algorithm and have the same view of the world. The $n$ robots are divided into $k$ destinations forming $k$ points of multiplicity $\alpha = (\alpha_1,\alpha_2,\ldots, \alpha_k)$, with $|\alpha| = \sum_{i=1}^k\alpha_i = n$. Consider the random variable $X$ that equals $\alpha\in\N^k$ if and only if the robots in $P$ are divided in $k$ points of multiplicity $\alpha_1$, $\alpha_2$, $\ldots$ and $\alpha_k$.

Recall the recursive inequality (\ref{eq:recursive lower bound for RBC}):
\begin{equation}\label{rec}
B(n) \geq n\log(k) +
 \sum_{\alpha\in\mathbb{N}^k,\quad
 |\alpha|=n} \Prob(X=\alpha)\sum_{i=1}^kB(\alpha_i)
\end{equation}

We assume in the remainder of the proof that $\mathcal{A}$ is optimal in terms of random bits (so that $B(n)\geq B(n-1)$). In order to prove the theorem, it is sufficient to show that $B(n)\geq n\log(n)$ for that algorithm. 

We already know that $B(2) \geq 2 = 2\log(2)$. Furthermore $k\geq 2$. We suppose now that for all $m$, with $2 \leq m < n$, we have $B(m)\geq m\log(m)$.

We now bound each right sum of (\ref{rec}):\\
\renewcommand{\arraystretch}{1.5}
For the case $\|\alpha\|_{\infty}=n$, since $\mathcal{A}$  is optimal, we have:
\begin{equation}\label{eq:first inequation}
\begin{array}{rl}
\sum_{i=1}^kB(\alpha_i) &= B(n)\geq B(n-1)\geq (n-1)\log(n-1)\\
&\geq n\log\left(\frac{n}{k}\right) + n\log\left(\frac{k}{n}(n-1)\right) - \log(n-1)\\
&= n\log\left(\frac{n}{k}\right) + \log\left(\frac{k}{n}\frac{(n-1)^n}{n-1}\right)\\
&= n\log\left(\frac{n}{k}\right) + \log\left(\frac{k}{n}(n-1)^{n-1}\right)\\
&\geq n\log\left(\frac{n}{k}\right) + \log\left(\frac{k}{n}(n-1)\right)\\
&= n\log\left(\frac{n}{k}\right) + \log\left(k-\frac{k}{n}\right)
\geq n\log\left(\frac{n}{k}\right)\\
\end{array}
\end{equation}
The last inequality is true because $2 \leq k < n \Longrightarrow k-\frac{k}{n}\geq 1 $.

Let $\alpha\in\mathbb{N}^k$ with $|\alpha|=n$ be a possible distribution.
If $\alpha_{i_0} = n$ for some $i_0$, then (\ref{eq:first inequation}) implies  $\sum_{i=1}^kB(\alpha_i) = B(n)\geq n\log(\frac{n}{k})$. Else, there exists $r$, with $1< r \leq k$, and $r$ integers $1\leq j_1$, $j_2$, $\ldots$, $j_r\leq k$ such that  $\sum_{i = 1}^r \alpha_{j_i} = n$ and $\alpha_{j_i}\neq0$ for all $1\leq i\leq r$. Those are the $r$ indexes of the $r$ non-zero coordinates of $\alpha$. We apply the induction hypothesis to each non-zero term of the sum:
\[\begin{array}{rl}
\sum_{i=1}^kB(\alpha_i)&= \sum_{i=1}^rB(\alpha_{j_i})
\geq \sum_{i=1}^r\alpha_{j_i}\log(\alpha_{j_i})
= r\sum_{i=1}^r\frac{1}{r}\alpha_{j_i}\log(\alpha_{j_i})\\
&\geq r\left(\sum_{i=1}^r\frac{1}{r}\alpha_{j_i}\right)\log\left(\sum_{i=1}^r\frac{1}{r}\alpha_{j_i}\right)
=n\log\left(\frac{n}{r}\right)\geq n\log\left(\frac{n}{k}\right)\\
\end{array}\]
The penultimate inequality results from the convexity of function $x\longmapsto x\log(x)$. Thus we lower bound each term of the recursive formula (\ref{rec}). This gives:
\[
\begin{array}{rl}
B(n) &\geq n\log(k) + \sum_{\alpha\in\mathbb{N}^k,\quad
 |\alpha|=n} 
\Prob(X=\alpha) n\log\left(\frac{n}{k}\right)\\


&\geq n\log(k) +  n\log\left(\frac{n}{k}\right) = n\log(k) +  n\log(n) - n\log(k)  = n\log(n)
\end{array}
\]
We have shown $B(n) \geq n\log(n)$ for all $n\geq 2$.
\end{proof}

\begin{theorem}\label{the:lower bound for RBC}
Let $\mathcal{A}$ be a scattering algorithm. The expected number of random bits needed to scatter $n$ robots is greater that $n\log(n)$
\end{theorem}

\begin{proof}
Let $\mathcal{A}$ be a scattering algorithm. If $\mathcal{A}$ is a canonical scattering algorithm, then the previous Lemma implies the Theorem. Else, there exist some points where robots come from two different origins. But, since $\forall \alpha_i, \beta_i$, $B(\alpha_i+\beta_i) \geq B(\alpha_i)+B(\beta_i)$, we still have the recursive formula (\ref{eq:recursive lower bound for RBC}) applied for both origins.
\end{proof}

\begin{corollary}
Algorithms defined in \cite{ClementDPIM10,DieudonneP09} (see section \ref{sec:canonical scattering algorithm}) are random bit optimal.
\end{corollary}

\subsection{Upper Bound.}

If an algorithm computes a set of distinct points, and if each robot chooses randomly a destination in this set, then robots must scatter. Moreover if the cardinality of the chosen set is bounded, then the expected number of random bits may be bounded too. We now prove that if the destination set cardinality is bounded by a polynomial in $|C|$, the random bit complexity of the algorithm is $O(n\log(n))$ (that is, optimal).

We start with three technical lemmas. Lemma \ref{lemma-probability-alpha-bound} helps us to bound the maximum multiplicity obtained after a round, when $n$ robots are randomly distributed among $k$ possible destinations. Let $\Omega$ be the universe of the experiment of randomly and uniformly distributing $n$ robots among $k$ possible destinations.
\input{link-with-balls-into-bins.tex}

\begin{lemma} \label{lemma-if-k>2m^2-multiplicity-is_one}
    Let $X:\;\Omega\mapsto \N^{k}$ be the random variable that gives the distribution of $n$ robots among $k$ destinations. If $k\geq 2n^2$, then:
\[ 
\mathbb{P}(\|X\|_{\infty} > 1) \leq \dfrac{1}{2}\]
\end{lemma}
\begin{proof}
    
We have:
 \[ 
    \renewcommand*{\arraystretch}{2.2}
\begin{array}{rl}
\mathbb{P}(\|X\|_{\infty} = 1) &= \left(\dfrac{k-1}{k}\right)\left(\dfrac{k-2}{k}\right)\ldots\left(\dfrac{k-n}{k}\right)\\
&\geq\left(1-\dfrac{n}{k}\right)^n\geq\left(1-\dfrac{n^2}{k}\right)\geq\dfrac{1}{2}
\end{array}
\]   
\end{proof}

\begin{lemma}\label{lemma-probability-alpha-bound}
Let $X:\Omega\mapsto \N^{k}$ be the random variable that gives the distribution of $n$ robots among $k$ destinations (uniformly at random). If $k\leq An^K$ with $A,K\in\mathbb{N}$, then there exists $N_{A,K}\in\mathbb{N}$ (that depends only on $A$ and $K$), such that, for all $n\geq N_{A,K}$:
\[
\mathbb{P}\left (\|X\|_{\infty} >\frac{n}{k}(1+k^\xi)\right ) \leq \dfrac{1}{2}\quad\text{with }\xi = 1-\frac{1}{K+1}
\]
\end{lemma}
\begin{proof}
Recall that:
\[
\forall\alpha\in\N^k,\;|\alpha|=n\qquad\mathbb{P}(X=\alpha) = 
\left(\begin{array}{c}n\\\alpha\end{array}\right)\left(\dfrac{1}{k}\right)^n
\]
Where:
\[
\left(\begin{array}{c}n\\\alpha\end{array}\right)= \dfrac{n!}{\alpha_1! \ldots\alpha_k!}=
\left(\begin{array}{c}n\\\alpha_1\end{array}\right)
\left(\begin{array}{c}n-\alpha_1\\\alpha_2\end{array}\right)
\ldots
\left(\begin{array}{c}n-\alpha_1-\alpha_2\;\ldots-\alpha_{k-2}\\\alpha_{k-1}\end{array}\right)
\]
The main result we use in the sequel is this equality, for $r\leq n$:
\[
\begin{array}{ll}
\sum_{\begin{array}{c}
\alpha\in\mathbb{N}^k\\|\alpha|=n\\\alpha_1=r
\end{array}}
\mathbb{P}(X = \alpha) &= 
\sum_{\begin{array}{c}
\alpha\in\mathbb{N}^k\\|\alpha|=n\\\alpha_1=r
\end{array}}\left(\begin{array}{c}n\\\alpha\end{array}\right)
\left(\dfrac{1}{k}\right)^n \\
&=
\left(\begin{array}{c}n\\r\end{array}\right)
\left(\dfrac{1}{k}\right)^n
\sum_{\begin{array}{c}
\alpha\in\mathbb{N}^{k-1}\\|\alpha|=n-r
\end{array}}
\left(\begin{array}{c}n-r\\\alpha\end{array}\right)\\
&= 
\left(\begin{array}{c}n\\r\end{array}\right)
\left(\dfrac{1}{k}\right)^r
\left(\dfrac{1}{k}\right)^{n-r}
\left(k-1\right)^{n-r}\\
&= 
\left(\begin{array}{c}n\\r\end{array}\right)
\left(\dfrac{1}{k}\right)^r
\left(1-\dfrac{1}{k}\right)^{n-r}\\
&= 
\mathbb{P}(B_{n,\frac{1}{k}} = r )
\end{array}
\]
Let $\xi = 1 - \frac{1}{K+1}$. The Chernoff bound\footnote{Let $X$ be random variables following a Binomial distribution $B(n,p)$. Then:
for any $\delta>0$, $\Prob(X\geq np(1+\delta)) \leq e^{\frac{-\delta^2np}{2+\delta}}$ (see \cite{angluin1979fast} in the appendix references)} gives:
\[
\mathbb{P}(B_{n,\frac{1}{k}} > \frac{n}{k}(1+k^\xi) ) \leq e^{-\dfrac{k^{2\xi}n}{(2+k^\xi)k}} =
e^{-\dfrac{k^{\xi-1}n}{2/k^\xi+1} }
\leq e^{-\dfrac{k^\frac{-1}{K+1}n}{3}}
\]
But since $k = O(n^K) = o(n^{K+1/2})$, there exists $N_1$ from which $k<n^{K+1/2}$ \emph{i.e}:
\[
\forall n\geq N_1,\qquad
k^{\frac{-1}{K+1}} > n^{\frac{-(K+1/2)}{K+1}}
\]

So 
\[
\forall n\geq N_1,\qquad
\mathbb{P}(B_{n,\frac{1}{k}} > \frac{n}{k}(1+k^\xi) ) \leq
e^{-\dfrac{k^\frac{-1}{K+1}n}{3}}  \leq e^{-\dfrac{n^{\frac{-(K+1/2)}{K+1} + 1}}{3} }= e^{-\dfrac{n^{\frac{1/2}{K+2}}}{3}}
\]

And then
\[
\forall n\geq N_1,\qquad
\mathbb{P}(\|X\|_{\infty} >\frac{n}{k}(1+k^\xi)) \leq ke^{-\dfrac{n^{\frac{1/2}{K+1}}}{3}}\]
Since $k = O(n^K)$ the lemma follows with $N_{A,K} \geq N_1$.
\end{proof}

From now on we use, for a canonical scattering algorithm $\mathcal{A}$, the notation $W(n,N)$ for the largest expected number of random bits used by $\mathcal{A}$ to scatter $n$ robots gathered in a point $P$ in a configuration of $N$ robots (see Equation (\ref{eq: B an W definition})). 

\begin{lemma}\label{lemma-bound-expecting-for-bounded-number-of-robots}
Let $\mathcal{A}$ be a canonical scattering algorithm with a destination function that satisfies: $\exists K\in \mathbb{N}, k_\mathcal{A}(C,P) = O(|C|^{K})$. Then, for all $\mathcal{N}\in\mathbb{N}^*$, there exists $R\in\mathbb{R}$ such that: 
\[
\forall n\leq\mathcal{N},\;\forall N\geq n, \qquad W(n, N)\leq R\log(N)
\]
\end{lemma}
\begin{proof}
Let $N\in\N$, and $C$ be a $N$-robot configuration containing a point $P$ of multiplicity $n$. Since $k_\mathcal{A}(C,P) = O(N^{K})$, there exists $R_0$ such that $\log(k_\mathcal{A}(C,P)) \leq R_0\log(N)$. Moreover there exists $p\in\N$ such that for all $2\leq n\leq\mathcal{N}$ and for all $k\geq 2$, $n$ robots moving randomly toward $k$ possible destinations are split into at least two points with probability at least $\frac{1}{p}$. In other words, define $Y_{n,k}$, with $n\leq\mathcal{N},\;k\in\N$, the random variable that equals $0$ if all robots at the point $P$ of multiplicity $n$, moving randomly among $k$ possible destinations, are still gathered, and $1$ if not. Let $p$ such that $\frac{1}{p} \leq \Prob(Y_{\mathcal{N},2} = 1)$. Then we have
\[
\forall n\leq \mathcal{N},\;\forall k\geq 2\qquad\Prob(Y_{n,k}=1)\geq\frac{1}{p}
\]

So that the expected number of rounds needed to decrease the multiplicity of $P$ by one is $p$. Then we have:
\[
W(n,N) \leq pnR_0\log(N) + W(n-1,N)
\]
and recursively we have:
\[
W(n,N) \leq pn^2R_0\log(N)
\]
Since $n\leq \mathcal{N}$, with $R=p\mathcal{N}^2R_0$, we have:
\[
W(n,N) \leq W(\mathcal{N},N) \leq R\log(N)
\]
\end{proof}
\begin{theorem}\label{theorem-optimality-caracterization}
Let $\mathcal{A}$ be a canonical scattering algorithm with a destination function that satisfies: 
\[
    \exists K\in \mathbb{N},\quad k_\mathcal{A}(C,P) = O(|C|^{K})
\]
Then $\mathcal{A}$ is optimal in terms of random bits, \emph{i.e.} the expected number of random bits needed to scatter $n$ robots is $O(n\log(n))$.
\end{theorem}

The sketch of the proof is as follows. We use mathematical induction over the global number $N$ of robots ($N=|C|$) and the local number $n$ of robots located at position $P$.
We show that there exist $R$ and $R'$ such that:
\begin{equation}
\label{eq:upper bound indeuctive formula}
W(n,N) \leq nR\log(n) + nR'\log(N)
\end{equation}
Then for a configuration where all $n$ robots are gathered, we have $W(n,n) \leq n(R+R')\log(n) = O(n\log(n))$. Lemma \ref{lemma-bound-expecting-for-bounded-number-of-robots} is used for the base case. Indeed, for all $\mathcal{N}\in\N$, we can assign a value to $R'$ in order to make Equation (\ref{eq:upper bound indeuctive formula}) true for all $n\leq\mathcal{N}$ and for all $N\geq n$. Then, Lemma \ref{lemma-probability-alpha-bound} is used in the inductive step. When $n$ robots are randomly and uniformly distributed among $k$ possible destinations, there is a high probability that the distribution is almost fair, and a low probability that a large number of robots moves toward the same destination. Lemma \ref{lemma-probability-alpha-bound} indicates how fair the distribution can be with probability greater than $1/2$. Overall, two rounds (in expectation) are sufficient to have this ``almost fair" distribution. We then use the inductive hypothesis with the new points.

\begin{proof}
We show that there exist $R$ and $R'$ such that:
\begin{equation}
\label{eq:upper bound indeuctive formula apdx}
W(n,N) \leq nR\log(n) + nR'\log(N)
\end{equation}

We first have to define $R$ and $R'$.
\newcommand\nomathlinebreaks{\relpenalty=10000
                             \binoppenalty=10000 }
Since $k_{\mathcal{A}}(C,P) = O(|C|^K)$, there exists $R_0\geq K$ such that $2\log(k_{\mathcal{A}}(C,P)) \leq R_0\log(|C|)$. Let $\mathcal{N}=N_{2,2}$ defined in Lemma~\ref{lemma-probability-alpha-bound}. By Lemma~\ref{lemma-bound-expecting-for-bounded-number-of-robots}, there exists $R_1\in\R$ such that: $\forall n\leq \mathcal{N},\;\forall N\geq n$ :
\[W(n,N) \leq R_1\log(N)\]

Now, take $R \geq 74$, $R'\geq\max(R_0,R_1)$, so that the induction hypothesis (\ref{eq:upper bound indeuctive formula apdx}) is true for $n\leq\mathcal{N}$ and all $N\geq n$. We now let $n>\mathcal{N}$ and suppose that (\ref{eq:upper bound indeuctive formula apdx}) is true for $m<n$: 
\[
\forall m<n, \qquad
W(m,N) \leq mR\log(m) + mR'\log(N)
\]
We now have to show that this is true with $n$.

Let $(C,P)$ be such that $\E(Z_{C,P}) = W(n,N)$ and $k = k_\mathcal{A}(C,P)$ be the number of possible destinations computed by $\mathcal{A}$. 
We assume that cases $k=2, 3, 4, 5, 6, 7$ or $8$ can be done the same way (see Lemma \ref{lem:case k = 2} for the case $k=2$) maybe with a greater $R$. Thereby we suppose $k\geq 9$. 
Recall the recursive inequality (\ref{eq:recursive upper bound for RBC}):
\begin{equation}\label{eq:recursive formule for upper bound RBC}
W(n,N) \leq n\log(k) + \sum_{\alpha\in\mathbb{N}^k,\quad|\alpha|=n}\left(\dfrac{1}{k}\right)^n
\left(\begin{array}{c}n\\\alpha\end{array}\right)\sum_{i=1}^kW(\alpha_i,N)    
\end{equation}

\paragraph{If $k\geq 2n^2$.} Then $\Prob(\|X\|_{\infty} > 1) \leq \frac{1}{2}$ (Lemma \ref{lemma-if-k>2m^2-multiplicity-is_one}). So:
\[
W(n,N) \leq n\log(k) + \frac{1}{2}W(n,N)
\]
And then:
\[
W(n,N) \leq 2n\log(k) \leq nR'\log(N)
\]

\paragraph{If $k < 2n^2$.}

We now split the sum on the right hand of (\ref{eq:recursive formule for upper bound RBC}). One part, where the distribution is almost fair, the other part where there is a point that is the destination of an abnormally large number of robots. The maximum multiplicity tolerated is $M_k = \left\lfloor\frac{n}{k}(1+k^{2/3})\right\rfloor$. If $ \max_{i}(\alpha_i) > M_k$, then we bound the sum $\sum_{i=1}^kW(\alpha_i,N)$ by $W(n,N)$. Else, each $\alpha_i$ is less than $M_k$ and the worst distribution happens when there are most points with multiplicity $M_k$. Let $m_k$ be the maximum number of points with multiplicity $M_k$. We have:
\[\sum_{i=1}^kW(\alpha_i,N)\leq m_kW(M_k,N)+W(n-m_kM_k,N)\]
So that the split of (\ref{eq:recursive formule for upper bound RBC}) gives:
\[W(n,N) \leq n\log(k) + 
pW(n,N)
+
(1-p)\left(m_kW(M_k,N)+W(n-m_kM_k,N) \right)
\]
With 
\[p =\Prob(\|X\|_{\infty} > M_k)\]
Since $n>\mathcal{N} = N_{2,2}$, by Lemma~\ref{lemma-probability-alpha-bound} we have $p<\frac{1}{2}$.
So:\\
\[W(n,N) \leq n\log(k) + \frac{1}{2}W(n,N) + \frac{1}{2}\left(m_k W\left(M_k,N\right) + W\left(n-m_kM_k,N\right)\right)\]
recursively we can show:
\[W(n,N) \leq 2n\log(k) +  m_k W\left(M_k,N\right) + W\left(n-m_kM_k,N\right)\]
If $M_k\leq 1$, then we have $W(n,N) \leq 2n\log(2n^2) = 4n(\log(n) + 1) $\\
\\
$\bullet$ Consider the case  $M_k > 1$ and $m_kM_k \leq n-1$.\\
The induction hypothesis implies:
\[
\begin{array}{rl}
    W(n,N) \leq 2n\log(k)+& nR'\log(N) \\
    +& R\left[m_kM_k\log\left(M_k\right)  + \left(n-m_kM_k\right)\log\left(n-m_kM_k\right)\right]
\end{array}
\]
By the concavity of $x\mapsto \log(x)$ we deduce:
\[W(n,N) \leq nR'\log(N) +  2n\log(k) +  nR\log\left(\dfrac{m_kM_k^2 + (n-m_kM_k)^2}{n}\right)
\]
Firstly we bound $m_kM_k$ by $n$, secondly $M_k$ and $n-m_kM_k$ by $\frac{n}{k}(1+k^{2/3})$:
\[W(n,N) \leq  nR'\log(N) + 2n\log(k) +  nR\log\left(\frac{\frac{n^2}{k}(1+k^{2/3}) + \frac{n^2}{k^2}(1+k^{2/3})^2}{n}\right)
\]
\[
\begin{array}{rl}
    W(n,N) \leq  &nR'\log(N) + nR\log(n) \\&+ n\left[ 2\log(k) +  R\log\left(\frac{k(1+k^{2/3}) + (1+k^{2/3})^2}{k^2}\right)\right]
\end{array}
\]
But since $R \geq 74$, then for all $k\geq 9$:
\[ R\geq \dfrac{\log(k^2)}{\log\left(\frac{k^2}{k(1+k^{2/3}) + (1+k^{2/3})^2}\right)} \quad \text{ and }\quad\log\left(\frac{k^2}{k(1+k^{2/3}) + (1+k^{2/3})^2}\right)>0 \]
So that
\[2\log(k)+  R\log\left(\frac{k(1+k^{2/3}) + (1+k^{2/3})^2}{k^2}\right)\leq 0 \]
And we obtain:
\[W(n,N) \leq  nR'\log(N) +  nR\log(n)\]

$\bullet$ There remains the case $M_k > 1$ and $m_kM_k = n$
We bound $ W(n,N)$ in the same way:
\[W(n,N) \leq 2n\log(k) + nR'\log(N) +   nR\log\left(\frac{n}{k}(1+k^{2/3})\right) \]
\[W(n,N) \leq nR'\log(N) +  nR\log\left(n\right) + n\left[2\log(k) + R\log\left(\frac{1+k^{2/3}}{k}\right)\right] \]
But since $R \geq 74$, then for all $k\geq 9$:
\[ R\geq\frac{\log(k^2)}{\log\left(\frac{k}{1+k^{2/3}}\right)} \qquad \text{and }\qquad\log\left(\frac{k}{1+k^{2/3}}\right)>0\]
And again, we have:
\[W(n,N) \leq nR'\log(N) +  nR\log\left(n\right)\]
\end{proof}

\section{Time Complexity without Strong Multiplicity Detection}\label{sec:time complexity without strong multiplicity detection}

In this section, we investigate the time complexity (that is, the expected scattering time) of scattering algorithms that do not use strong multiplicity detection. We already know that global strong multiplicity detection 
(that permit to compute the number $n$ of robots) 
enables $O(1)$ expected scattering time (see the algorithm of Clement \emph{et al.} in the previous section). That bound obviously still holds if only \emph{local} strong multiplicity detection is available (their scattering algorithm is canonical, so different multiplicity points are independent). There remains the case of weaker forms of multiplicity detection (that is, local and global weak multiplicity, or no multiplicity detection whatsoever). We essentially show that with respect to time complexity, weak multiplicity detection does not help. 
Without strong multiplicity detection, we show that: for any algorithm, the optimal expected $O(1)$ cannot be achieved; for random bit optimal algorithms, at least $\Omega(\log \log n)$ expected rounds are necessary. On the positive side, we present a family of scattering algorithms that do \emph{not} use multiplicity detection yet can achieve arbitrarily fast (yet not constant) expected time. Of particular interest in this family is a scattering algorithm that is both random bit optimal scattering protocol and scatters $n$ robots in $O(\log \log n)$ expected rounds.

\begin{theorem}\label{thm:O(1) is impossible}
There exists no scattering algorithm with $O(1)$ expected rou\-nds complexity that uses only global weak multiplicity detection.
\end{theorem}

\begin{proof}
Suppose that there exists $E\in\N$, such that for every $n\in\N$, the expected number of rounds needed by $\mathcal{A}$  to scatter $n$ robots is less than $E$.
Let $u\in\N$, and $\mathscr{P}$ be a set of $u$ points. Consider the equivalence relation $\backsim$ over the set of configurations $C$ that satisfy $U(C)\subset \mathscr{P}$ such that $C\backsim C'$, if $C$ and $C'$ cannot be distinguished with only the weak multiplicity detection. There is a finite number of equivalence classes, so the image of $k_{\mathcal{A}}$ is finite. So, after $E$ rounds there is a maximum number of points where robots can lie, and if $n$ is greater than that number, no $n$-robots configuration can be scattered in $E$ rounds. A contradiction.
\end{proof}

The following lemmas are used in algorithm \ref{algo-scattering-weak-multiplicity} and in theorem \ref{th:SA_f converge in O(f(n)) rounds}.

\begin{lemma}\label{lemma-probability-alpha-bounded-by-1/k}
    Let $X$ be the random variable that gives the distribution of $m$ robots among $k$ destinations (uniformly at random). There exists $\mathcal{N}$, such that for all $m>\mathcal{N}$ and $k \leq 8m^3$:
\[ 
\mathbb{P}(\|X\|_{\infty} >\frac{m}{k}(1+k^\frac{3}{4})) \leq \dfrac{1}{2k}
\]
\end{lemma}
\begin{proof}
    As in Lemma \ref{lemma-probability-alpha-bound}, with $K = 3$ and $\xi = 3/4$,
we have:
\[
\mathbb{P}(\|X\|_{\infty} >\frac{m}{k}(1+k^{3/4})) \leq ke^{-\dfrac{m^{\frac{1}{8}}}{3}} \leq \frac{1}{k}k^2e^{-\dfrac{m^{\frac{1}{8}}}{3}}
\]
And there exists $\mathcal{N}$ such that for all $m>\mathcal{N}$ and $k \leq 8m^3$:
\[k^2e^{-\dfrac{m^{\frac{1}{8}}}{3}}<\frac{1}{2}\]
\end{proof}

\begin{lemma} \label{lemma-if-k>8m^3-multiplicity-is_one}
    Let $X$ be the random variable that gives the distribution of $m$ robots among $k$ destinations. If $k>8m^3$, then:
\[ 
\mathbb{P}(\|X\|_{\infty} > 1) \leq \dfrac{1}{2k^{1/3}}\]
\end{lemma}
\begin{proof}
As in Lemma~\ref{lemma-if-k>2m^2-multiplicity-is_one}, we have:
 \[ 
\begin{array}{rl}
\mathbb{P}(\|X\|_{\infty} = 1) &
\geq\left(1-\dfrac{m}{k}\right)^m\geq\left(1-\dfrac{m^2}{k}\right)\geq\left(1-\dfrac{1}{2k^{1/3}}\right)
\end{array}
\]   
\end{proof}

\begin{lemma}\label{lem:after one round, multiplicity of a point is divided by x}
Let $P$ a point where lie $m$ robots. For all $u\in\N,\; x\in\N$ we have: after a random and uniform distribution of robots at $P$ among $k = \max(16x^4;u^3;8\mathcal{N}^3)$ possible destinations, robots are divided into points of multiplicity $1$ or less than $m/x$ with probability at least $1-\frac{1}{2u}$.
\end{lemma}
\begin{proof}

If $k>8m^3$, by Lemma \ref{lemma-if-k>8m^3-multiplicity-is_one}, we have $$\Prob(\|X\|_\infty > 1)< \frac{1}{2k^{1/3}} < \frac{1}{2u} $$
Else, $k\leq 8m^3$ and since $k>8\mathcal{N}^3$ we have $m > \mathcal{N}$ and so by Lemma \ref{lemma-probability-alpha-bounded-by-1/k}:
$$\Prob(\|X\|_\infty > \frac{m}{k}(1+k^{3/4})) < \frac{1}{2k}< \frac{1}{2u}$$
But since $k>16x^4$, we have:$$\frac{1}{k}(1+k^{3/4})\leq \frac{1}{x}$$
And then: $$\Prob(\|X\|_\infty > \frac{m}{x}) < \frac{1}{2k}< \frac{1}{2u}$$
\end{proof}

\begin{lemma}\label{thm:multiplicity divided by x with constant probability}
Let $C$ be a configuration with $n$ robots organized in $u$ points of multiplicity at most $m$ (\textit{i.e.}, $U(C) = \{P_1,\;P_2,\;\ldots\;,\;P_u\}$). Let $x\in\N$. If there exists $u$ disjoint sets of  $k = \max(16x^4;u^3;8\mathcal{N}^3)$ points $D_1$, $D_2, \ldots , D_u$, such that all robots in $P_i$ are randomly distributed among points in $D_i$. Then the maximum multiplicity of the resulting configuration is $1$ or less than $m/x$ with probability at least $\frac{1}{2}$.
\end{lemma}
\begin{proof}
 Let $U(C) = \{P_1,\;P_2,\;\ldots\;,\;P_u\}$. We define the indicator random variable $Z_i$ as follows: $Z_i = 1$ if all robots located at the same point $P_i$ are located after one round on points of multiplicity either $1$ or less than $\dfrac{m}{x}$. $Z_i = 0$ otherwise. Notice that $\{Z_1, Z_2, \ldots, Z_u\}$ are mutually independent because the destinations sets $D_i$ are disjoint i.e. no two robots from different points ever reach the same position.

Since for all $i$, all robots at $P_i$ are randomly distributed among $k$ possible destinations, by Lemma~\ref{lem:after one round, multiplicity of a point is divided by x} we have: 
\[
\Prob(Z_i = 1) \geq 1- \dfrac{1}{2u}
\]
And we get:
\[
\Prob(\bigwedge_{i=1}^u Z_i = 1) 
\geq \left( 1- \dfrac{1}{2u}\right)^u
\geq 1- \dfrac{u}{2u}
 = \dfrac{1}{2}
\]
\end{proof}

Let $\F = \{f : \N\mapsto\N\;|\;f \text{ is increasing and surjective}\}$. Let $f\in\F$. We define $f^{-1}$ as the maximum of the inverse function: \emph{i.e.} $f^{-1}(y) = \max\{x\;;\;f(x) = y\}$. Since $f\in\F$, $f$ is not bounded and $f^{-1}:\N\mapsto\N$ is well defined, increasing and diverging. Moreover we have $f^{-1}(1)>0$.

Given a function $f \in\F$, we now define Algorithm $\SA_f$ (see Algorithm~\ref{algo-functionnal-scattering-weak-multiplicity}) that converges in $O(f(n))$ rounds in expectation (see Theorem \ref{th:SA_f converge in O(f(n)) rounds}).

\begin{algorithm}[H]
\label{algo-functionnal-scattering-weak-multiplicity}
Compute the Voronoï diagram of the observed configuration\\
Let $u = |U(C)|$ $\,$
and $\,x = f^{-1}(f(u)+1)$\\
Let $k = \max(8 \mathcal{N}^3, 16x^4, u^3)$ $\qquad$ with $\mathcal{N}$ given by Lemma~\ref{lemma-probability-alpha-bounded-by-1/k}\\
Let $Pos$ be a set of $k$ distinct positions in the Voronoï cell where $r$ is located  \\
Move toward a position in $Pos$ chosen uniformly at random.
\caption{\textbf{$\SA_f$}: Scattering algorithm executed by robot $r$. No multiplicity detection}
\end{algorithm}

$\SA_f$ is a canonical scattering algorithm under the FSYNC and SSYNC models. To construct the set of possible destinations, it executes the procedure given by the previous lemma with $x=f^{-1}(f(u)+1)$ where $u=|U(C)|$. Thus, if $m$ is the maximum multiplicity of a given configuration, then after one execution of $\SA_f$, the maximum multiplicity is $max(1, m/f^{-1}(f(u)+1))$ with probability at least $1/2$.


\begin{theorem}
    \label{th:SA_f converge in O(f(n)) rounds}
Let $f\in\F$. $\SA_f$ is an canonical scattering algorithm, which scatters $n$ robots in $O(f(n))$ rounds in expectation. 
\end{theorem}

The sketch of the proof is as follows. We first show that after $2i$ rounds in expectation, the maximum multiplicity of every point is less than $n/f^{-1}(i)$. Indeed we use Lemma~\ref{thm:multiplicity divided by x with constant probability} with $x = f^{-1}(f(n)+1)$. So that the expected number of rounds of an execution is less than $2f(n)$.

\begin{proof}
Given a function $f\in\F$, we define the following sequence of states (that depends on $f$, but we omit $f$ for clarity): depending on the maximum multiplicity $m$ we say that a $n$-robots configuration is in a state $i$ if:\\
$\bullet$ $ i = 0 $ and $m > n/f^{-1}(1)$.\\
$\bullet$ $1 \leq i\leq f(n)$ and $n/f^{-1}(i+1)< m \leq n/f^{-1}(i)$\\
$\bullet $ $i > i_{\max} =  f(n)$ and $m=1$ i.e. the configuration is scattered.

Since $\SA_f$ is a canonical algorithm, the multiplicity of points never increases. Since $f$ is increasing and diverging, the sequence of configurations of an executions of the algorithm is an increasing sequence of $r$ states $i_1\leq i_2\leq $ $\ldots$ $\leq i_r$, called \emph{the states of the execution}.

Suppose now that the configuration $C$ is in the state $i$ ( $i\neq i_{\max}$). Let $m$ be the maximum multiplicity and $u=|U(C)|$. We have $m\leq n/f^{-1}(i)$.
Since $mu\geq n$ and $m\leq n/f^{-1}(i)$, we have: $$f^{-1}(i)\leq \dfrac{n}{m}\leq u$$.

Then, by applying $f$ to each member, we have $i \leq f(u)$ and then: $$f^{-1}(i+1) \leq f^{-1}(f(u)+1)$$
Such that: 
\begin{equation}\label{eq:upper bound the maximum multiplicity after one execution}
\dfrac{m}{f^{-1}(f(u)+1)}\leq \dfrac{m}{f^{-1}(i+1)}\leq \dfrac{n}{f^{-1}(i+1)}
\end{equation}

By Lemma~\ref{thm:multiplicity divided by x with constant probability} and with (\ref{eq:upper bound the maximum multiplicity after one execution}): after one execution of $\A_f$, the probability that the maximum multiplicity is $1$ or is less than $\dfrac{n}{f^{-1}(i+1)}$ is at least $1/2$ i.e. the probability that the configuration state changes is at least $1/2$. So that the expected number of rounds needed for the state to change is at most $2$.\\

Let $X_i$ be the random variable, over all the possible executions of the algorithm $\SA_f$, that equals the number of rounds the execution stays in the state $i$.
We have just shown that $\E(X_i) \leq 2$ for all $i$.
Thus: $$\E(X_1 + X_2 + \ldots + X_{i_{\max}}) \leq 2i_{\max} = 2f(n)$$
Moreover, the expected number of rounds needed to scatter $n$ robots is less than the sum of the expected number of rounds that the algorithm stays at each state.
So that $2f(n)$ is an upper bound of  the expected number of rounds needed by $\SA_f$ to scatter $n$ robots.

\end{proof}

\begin{theorem}\label{th:RBC is theta(B(n))}
Let $\A$ be a canonical scattering algorithm such that at each activation, the number of possible destinations computed by each robot is the same \emph{i.e.} for every configuration $C$ that contains two points $P$ and $P'$, $k_{\A}(C,P) = k_{\A}(C,P')$. Let $B(n)$ be the maximum number of random bits used by a robot among all $n$-robots configuration. Then the random bit complexity of $\A$ is $\Theta(n\log(n) + nB(n))$.
\end{theorem}
\begin{proof}
    Firstly, it is clear that the random bit complexity is $\Omega(n\log(n))$ (Theorem \ref{the:lower bound for RBC}). More over if an execution start with the worst configuration, $\A$ uses $nB(n)$ random bits during the first round. So that the random bit complexity is $\Omega(n\log(n) + nB(n))$.

Secondly, by Theorem~\ref{the:lower bound for RBC}, we know that the expected number of random bits, used by all rounds where robots compute less than $2n^2$ possible destination is $O(n\log(n))$. Moreover the expected number of rounds where robots compute more than $2n^2$ possible destinations is less than $2$, so that the expected number of random bits used by all rounds that compute more than $2n^2$ possible destinations is $O(nB(n))$. And the random bit complexity is $O(n\log(n) + nB(n))$.
\end{proof}

A direct consequence of Theorem~\ref{th:RBC is theta(B(n))} is the following:

\begin{theorem}
    \label{thm:RBC of SA}
Let $f\in\F$. $\SA_f$ uses $\Theta(n\log(f^{-1}(f(n)+1)))$ random bits in expectation.
\end{theorem}
\begin{proof}

The maximum number of random bits used by robots executing $\SA_f$ is $B(n)=4\log(2f^{-1}(f(n-1)+1))$. That happens when the $n$ robots are split into $n-1$ points. Notice that, since $f^{-1}$ is an increasing function, we have $\log(n) = O(\log(f^{-1}(f(n)+1)))$.

By Theorem~\ref{th:RBC is theta(B(n))}, the random bit complexity of $\SA_f$ is $\Theta(n\log(f^{-1}(f(n)+1)))$.

\end{proof}

Note that the hypothesis that $f\in\F$ is not very restrictive. Indeed, if we want our algorithm to converge in $O(g(n))$ with $g$ a function that may not be increasing nor surjective but such that $\lim_{n\rightarrow +\infty}g(n)=+\infty$. We can define $f$ by : $f(0)=0$ and $\forall x>0,\; f(x) = \min\left(\max\left(g(x),f(x-1)\right), f(x-1)+1\right)$. So that $f\in\F$ and $O(f(n)) \subset O(g(n))$.

Now, our algorithm converges as fast as we want. We can try it with some convenient functions. For example, with $f = \log^*$, the algorithm $\SA_{\log^*}$ converges in $O(\log^*(n))$ rounds in expectation. Moreover, since
\footnote{we use the tetration notation : ${^{n}a} = \underbrace{a^{a^{\cdot^{\cdot^{a}}}}}_n$ i.e. $a$ exponentiated by itself, $n$ times.} 
$\log(f^{-1}(f(n)+1))= \log( ^{(\log^*(n)+1)}2)= ^{\log^*n}2=n$,
the resulting algorithm uses $O(n^2)$ random bits in expectation.
A faster algorithm can be obtained using the inverse Ackermann function $A^{-1}$ such that the time complexity of $\SA_{A^{-1}}$ is in $O(A^{-1}(n)) = o(log^* log^* log^* log^* n)$.

\paragraph{A Random Bit Optimal Algorithm.} If we want our algorithm $\SA_f$ to be random bit optimal, $f$ must satisfy: $n\log(f^{-1}(f(n)+1)) = O(n\log(n))$.

With $f = \log\composition\log$, we have: $n\log(f^{-1}(f(n)+1)) = n\log 2^{2^{\log\log(n)+1}} = 2n\log(n) = O(n\log(n))$. So that $\SA_{\log\composition\log}$ is random bits optimal and converge in $O(\log\log n)$ rounds in expectation. 
Also, the following theorem makes $\SA_{\log\composition\log}$ optimal in time.

\begin{theorem}\label{lower-bound-for-optimal-scattering-with-weak-multiplicity-detection}
There exists no random bit optimal scattering algorithm with $o(\log(\log(n)))$ expected rounds complexity that uses only global weak multiplicity detection.
\end{theorem}

\begin{proof}
Let $\mathcal{A}$ be a canonical scattering algorithm that uses weak multiplicity detection, and is random bit optimal.
Then, there exists $K$ such that $k_{\mathcal{A}}(C,P) = O(|C|^K)$, where  $C$ is a configuration containing a point $P$. Since $\mathcal{A}$ does not know $|C|$, $\mathcal{A}$ might not know whether $|C| > 2|U(C)|$. Then we can suppose that $K$ is such that $k_{\mathcal{A}}(C,P) = O(|U(C)|^K)$. Since $|U(C)|\geq 1$, this is equivalent to say that there exist $B$ such that $k_{\mathcal{A}}(C,P) \leq B|U(C)|^K$. Indeed $k_{\mathcal{A}}(C,P) = O(|U(C)|^K)$ implies:
\[
\begin{array}{rll}
 & \exists B_0,\; C_0,\;\forall |C|>&|C_0|, k_{\mathcal{A}}(C,P)\leq B_0 |U(C)|^K\\
\Longrightarrow& \exists B_0,\; C_0,\;\forall C,& k_{\mathcal{A}}(C,P)\leq B_0|U(C)|^K+\max_{|C|\leq |C_0|}\left(k_{\mathcal{A}}(C,P)\right)
\\ \Longrightarrow& \exists B_0,\; C_0,\;\forall C,& k_{\mathcal{A}}(C,P)\leq |U(C)|^K\left(B_0+\max_{|C|\leq |C_0|}\left(k_{\mathcal{A}}(C,P)\right)\right)
\\ \Longrightarrow& \exists B,\;\forall C,& k_{\mathcal{A}}(C,P)\leq B|U(C)|
\end{array}
\]

 So the maximum number of points in which at least one robot lies after one round is $n_1 = B$. After two rounds, robots are split into $n_2 = Bn_1^K =   B^{K+1}$ points at most. After $3$ rounds : $n_3=Bn_2^K=B^{K^2+K+1}$. After $r$ rounds we have:
\[n_r = B^{K^{r-1}+K^{r-2}+\ldots+ 1} \leq B^{K^{r}}\]

Suppose that, for all $n\in\N$, the expected number of rounds needed by $\mathcal{A}$ to scatter $n$ robots is less than $\varphi(n)$.
If $\varphi(n) = o(\log(\log(n)))$ we have:
\[
\begin{array}{rl}
\log(\log(B))+\log(K)\varphi(n) &= o(\log(\log(n)))\\
\log(B)K^{\varphi(n)} &= o(\log(n))\\
B^{K^{\varphi(n)}} &= o(n)\\
\end{array}
\]
So there exists $n_0\in\N$ such that $B^{K^{\varphi(n_0)}} < n_0$.
Since after $\varphi(n)$ rounds, robots are split into $B^{K^{\varphi(n)}}$ points at most, then, after $\varphi(n_0)$ rounds, $n_0$ robots cannot be scattered and the expected number of rounds cannot be less than $\varphi(n_0)$.
\end{proof}

\begin{corollary}
    $\SA_{\log\composition\log}$ is optimal for both time and random bit complexity.
\end{corollary}

The following table summarizes the dependency between time complexity and multiplicity detection.

\begin{tabular}{|c|c|c|}
\hline &Optimal time&Optimal time complexity \\
Multiplicity detection&complexity&for random bit\\
&&optimal algorithm
\\\hline Strong global or local&$O(1)$&$O(1)$
\\\hline Weak global or local&$\forall f,\;O(f(n))$&$O(\log\log(n))$
\\\hline No multiplicity detection&$\forall f,\;O(f(n))$&$O(\log\log(n))$
\\\hline
\end{tabular}
 
\section{Concluding Remarks}

We investigated the random bit complexity of mobile robot scattering and gave necessary and sufficient conditions for both (expected) random bit complexity and time complexity. It turns out that multiplicity detection plays an important role in the expected time complexity ($O(1)$ expected time can be achieved with strong multiplicity detection, while $\Theta(\log \log n)$ expected time complexity is optimal in the case of weak or no multiplicity detection) for the class of random bit optimal algorithms. 

We also found out that without strong multiplicity detection, even if the time complexity $O(1)$ is not reachable, there exist scattering algorithms that converge as fast as needed (yet not in expected constant time). Indeed our algorithms can have a time complexity of $O(f(n))$, for every increasing and surjective $f$, and expect $\Theta(n\log(f^{-1}(f(n)+1)))$ random bits in return.

An interesting remaining open question would be to prove whether our algorithms can be extended to the ASYNC model.

\bibliographystyle{plain}
\bibliography{Random_Bits_Complexity}

\appendix

\section{The case $k=2$}
\begin{lemma}\label{sec:case k = 2}\label{lem:case k = 2}Let $\A$ be a canonical scattering algorithm.
Let $n\in\N$. Let $(C,P)$ be such that $\E(Z_{C,P}) = W(n,N)$ and $k = k_\mathcal{A}(C,P)$ be the number of possible destinations computed by $\mathcal{A}$. If $k=2$ and if there exist $R\geq 17$ and $R'$ such that:
$$\forall m<n,\;\forall N\geq m,\qquad W(m,N)\leq mR\log(m)+mR'\log(N)$$
 Then: 
$$\forall N\geq n, \qquad W(n,N)\leq nR\log(n)+nR'\log(N)$$
\end{lemma}
\begin{proof}

If $n$ gathered robots split up, they form two groups of multiplicity $m$ and $n-m$ with probability $\frac{C^m_n}{2^n}$.
Hence we have the recursion formula:
\[\begin{array}{rl}
W(n,N) &= n + \frac{1}{2^n}\sum_{m=0}^{n}C^m_n(W(m,N) + W(n-m,N) )
\end{array}
\]
\[W(n,N) = n + \left(\sum_{m=\left\lceil\frac{n}{4}\right\rceil}^{\left\lfloor\frac{3n}{4}\right\rfloor}\dfrac{C^m_n}{2^n}(W(m,N) + W(n-m,N))\right)\phantom{-----}\]
\[\phantom{---} + 2\left(\sum_{m=\left\lfloor\frac{3n}{4}\right\rfloor+1}^{n}\dfrac{C^m_n}{2^n}(W(m,N) + W(n-m,N))\right)\]
\[W(n,N) \leq n + 
\left(W\left(\left\lfloor\frac{3n}{4}\right\rfloor,N\right) + W\left(\left\lceil\frac{n}{4}\right\rceil,N\right)\right)\left(\sum_{m=\left\lceil\frac{n}{4}\right\rceil}^{\left\lfloor\frac{3n}{4}\right\rfloor}\dfrac{C^m_n}{2^n}\right)\phantom{----}\]
\[\phantom{-------} + 
2W(n,N)\left(\sum_{m=\left\lfloor\frac{3n}{4}\right\rfloor+1}^{n}\dfrac{C^m_n}{2^n}\right)\]

According to the Chernoff bound, since $n > 16$, we have $P(B_{n,\frac{1}{2}}\geq \dfrac{3n}{4})\leq e^{-n/12} \leq e^{-17/8} < \dfrac{1}{4}$, so that:
\[W(n,N) \leq 2n + W\left(\left\lfloor\frac{3n}{4}\right\rfloor,N\right) + W\left(\left\lceil\frac{n}{4}\right\rceil,N\right) \]
\\
The induction hypothesis implies:
\\ 
\[W(n,N) \leq 2n +nR'\log(N) + R\left\lfloor\frac{3n}{4}\right\rfloor\log\left(\left\lfloor\frac{3n}{4}\right\rfloor\right) + R\left\lceil\frac{n}{4}\right\rceil\log\left(\left\lceil\frac{n}{4}\right\rceil\right)\]
\[W(n,N) \leq 2n +nR'\log(N) + Rn(\dfrac{\left\lfloor\frac{3n}{4}\right\rfloor}{n}\log\left(\left\lfloor\frac{3n}{4}\right\rfloor\right) + \dfrac{\left\lceil\frac{n}{4}\right\rceil}{n}\log\left(\left\lceil\frac{n}{4}\right\rceil\right)\]

\[W(n,N) \leq 2n +nR'\log(N) + Rn\left(\log\left(\dfrac{\left\lfloor\frac{3n}{4}\right\rfloor}{n}\left\lfloor\frac{3n}{4}\right\rfloor+\dfrac{\left\lceil\frac{n}{4}\right\rceil}{n}\left\lceil\frac{n}{4}\right\rceil\right)\right) \]
\[W(n,N) \leq 2n +nR'\log(N) + Rn\log\left(n\left(\frac{9}{16}+\frac{1}{16}+\frac{1}{2n}+\frac{1}{n^2}\right)\right) \]
\[W(n,N) \leq nR\log(n) +nR'\log(N)+ 2n + Rn\log\left(\frac{12}{16}\right) \]
And since $R\geq 17>\frac{2}{\log\left(\frac{16}{12}\right)}$, we have:
\[W(n,N) \leq nR\log(n) + nR'\log(N)\]

\end{proof}

\end{document}

%% file: link-with-naming-problem.tex
%
%

The scattering problem when robots know the total number of robots (\ie, when robots are capable of strong multiplicity detection) is closely related to the self-stabilizing \emph{Unique Naming Problem} \cite{unp}. Indeed, scattering two robots by choosing two different destinations is equivalent to choosing two different names. In fact, lemma \ref{lem:lower bound for RBC} can be deduced from this related topic (and even from other studies on probabilistic processes). Also, in the remaining of the paper, unique naming problem can give insight about the expected results. But, except in lemma  \ref{lem:lower bound for RBC}, proofs are strongly dependent on the oblivious and autonomous mobile robots model that contains assumptions, such as different multiplicity detection and obliviousness, so that our results cannot be directly deduced from the unique naming problem. That is why we also choose to prove lemma \ref{lem:lower bound for RBC} using our model to be more consistent with the remaining of the paper.

%% file: link-with-balls-into-bins.tex
%
%
This lemmas are related to previous studies called \emph{balls into bins} \cite{balls-into-bins}. 
Some of them, lemma \ref{lemma-if-k>2m^2-multiplicity-is_one} for instance, can be directly deduced from previous work but others differ by the approach. 
Indeed, the number of destinations (bins) depends on several parameters and in the section \ref{sec:time complexity without strong multiplicity detection}, the resulting probability needs to be bound by non-standard parameters (namely, the number $|U(C)|$ of observed robots). More precisely, for lemma \ref{lem:after one round, multiplicity of a point is divided by x}, the bound we derive is with high probability on $|U(C)|$ (contrary to the number on bins in previous results \cite{balls-into-bins}).
That is why we choose not to deduces those lemma from previous results.